\documentclass[journal]{IEEEtran}

\ifCLASSINFOpdf
\else
\fi

\usepackage{soul}
\usepackage{tikz}

\usepackage{hyperref}
\usepackage{amsmath}
\usepackage{amssymb}
\usepackage{subcaption} 
\usepackage[ruled, linesnumbered]{algorithm2e} % For algorithms

\usepackage{epic,eepic,graphics,graphicx,epsf}
\usepackage{url}
\usepackage{amsthm}
\usepackage{bm}
\usepackage{soul,xcolor}
\newtheorem{theorem}{Theorem}
\newtheorem{remark}{Remark}
\newtheorem{definition}{Definition}
\newtheorem{example}{Example}
\newtheorem{corollary}{Corollary}

\newcommand{\Exists}{\bm{\exists}\kern-0.6em\bm{\exists}}
\newcommand\copyrighttext{%
  \footnotesize This work has been submitted to the IEEE for possible publication. Copyright may be transferred without notice, after which this version may no longer be accessible.}
\newcommand\copyrightnotice{%
\begin{tikzpicture}[remember picture,overlay]
\node[anchor=south,yshift=10pt] at (current page.south) {\fbox{\parbox{\dimexpr\textwidth-\fboxsep-\fboxrule\relax}{\copyrighttext}}};
\end{tikzpicture}%
}
\begin{document}
\title{``ReLIC: Reduced Logic Inference for Composition'' for Quantifier
Elimination based Compositional Reasoning and Verification}
\author{
Hao~Ren,~\IEEEmembership{Member,~IEEE,}
Ratnesh~Kumar,~\IEEEmembership{Fellow,~IEEE,}
and~ Matthew~Clark,~\IEEEmembership{Member,~IEEE}
\thanks{H. Ren is with the Honeywell Aerospace Advanced Technology, Plymouth,
MN, 55441 USA. E-mail: hao.ren2@honeywell.com; R. Kumar is with the Department of Electrical \& Computer Engineering and Computer Science, Iowa State University, Ames, IA, 50010 USA. E-mail: rkumar@iastate.edu; M. Clark is with Galois, Dayton, OH, 45402, USA. E-mail: mattclark@galois.com.}
\thanks{This work was supported by the U.S. National Science Foundation under grants NSF-CCF-1331390, NSF-ECCS-1509420, NSF-PFI-1602089, and NSF-CSSI-2004766. }
}
\maketitle
\copyrightnotice
\begin{abstract}
The paper presents our research on quantifier elimination (QE) for compositional reasoning and verification. For compositional reasoning, QE provides the foundation of our approach, serving as the {\em calculus for composition} to derive the {\em strongest system-property} in a single step, from the given component atomic-properties and their interconnection relation. We first developed this framework for time-independent properties, and later extended it to time-dependent property composition. The extension requires, in addition, shifting the given properties along time to span the time horizon of interest, he least of which for the strongest system-property is no more than the total time horizons of the component level atomic-properties. The system-initial-condition is also composed from atomic-initial-conditions of the components the same way. It is used to verify a desired system-level property, alongside the derived strongest system-property, by way of induction. Our composition approach is uniform regardless of the composition types (cascade/parallel/feedback) for both time-dependent and time-independent properties.
%and also works for input-complete as well as input-incomplete properties. 
We develped a new prototype verifier named \texttt{ReLIC} (Reduced Logic Inference for Composition) that implements our above approaches. We demonstrated it through several illustrative and practical examples. Further, we advanced the $k$-induction based model-checking with QE capabilities, by formulating its base and inductive steps into QE problems where all the variables are universally quantified. Our integration of the QE solver \texttt{Redlog} with the $k$-induction based model-checking tool \texttt{JKind}, shows the successful solving of a non-linear problem that the SMT capable \texttt{JKind} failed to resolve. Finally, we also showcase the recent adoption of our approaches within an industrial V\&V tool suite for augmented static analysis of Simulink models and Deep Neural Networks (DNNs).
\end{abstract}


% \fi
\IEEEpeerreviewmaketitle

\section{Introduction}
\IEEEPARstart{D}{istributed} Cyber-Physical systems (CPSs) are integrations of computation, networking, and physical components. For instance, Unmanned Systems Autonomy Services (UxAS) system \cite{rasmussen2016field} integrates of a collection of services that communicates via a common messaging architecture. It is a framework to host flexible construction and deployment of software services that implement autonomy algorithmes on-board Unmanned Aerial Vehicles (UAVs) \cite{kingston2016automated} to enable autonomous functionalities. Successful missions of the UAVs depend on highly reliable design. One aspect of a dependable design is a model-based approach for testing/validation/verification of a distributed CPS. Within this paradigm, compositional reasoning is employed for the scalability of the model-based reasoning and analysis methods, and utilizes composition of component properties to derive the properties of a system composed of those components. Without the flexibility of compositional reasoning, a system often is designed with a modular architecture  \cite{schmidt2013verification}.

In the work presented here, we demonstrate that ``quantifier elimination'' serves as a calculas of property composition for compositional reasoning through the strongest system-property computation, that is further used to aid verification of postulated system properties. Quantifier elimination (QE) is a powerful technique for gaining insight into logic problems through simplification. The QE process essentially projects a logic formula to a lower dimension space of only its free variables. For instance, in real domain, $\exists x(y \ge x^2)$ can be reduced to its logical equivalence $y\ge 0$ (since $x^2$ must be either positive or zero), in which the dimension $x$, existentially quantified in the former formula, has been projected out. A theory is said to admit QE if there exists a quantifier-free equivalence for every formula in this theory. Examples of such are real closed field and an extension of  Presburger arithmetic (a first-order theory of the natural numbers with addition) \cite{tarski1998decision, enderton2001mathematical}. This suggests that QE is applicable to many practical problems in the real world.

%In this paper we propose a QE-based framework that performs property composition, and also introduced the notion/proof of the \textit{strongest system-property} derived from property composition. Given a system composed of $N$ components, wherein the ``atomic-property'' of the $i^\text{th}$ component is expressed as a predicate $\phi_i$ over its input-output variables, we show that the ``strongest system-property'' that can be claimed that the composed system satisfies, can be obtained by existentially quantifying the system's internal variables in the conjunct of component contracts, $\bigwedge_{i=1}^N (\phi_i)$ together with the additional constraints resulting from the component interconnection relations. Thereby we establish that QE serves as a foundation for property/contract composition. 

Our QE-based composition approach is unique in its uniformity regardless of the composition types (cascade/parallel/feedback) for both time-dependent and time-independent properties, and also works for input-complete as well as input-incomplete properties. Initially, each component is annotated with a given property, expressed as FOL formulae over component input and output variables. Once the {\em strongest system property} has been inferred employing our QE-based approach, next checking if the composed system satisfies a desired system-level  property, is reduced to checking if the aforementioned strongest system-property implies the desired system-level property, which can be formulated as a QE problem in itself. %The results of the paper apply to any distributed CPS such as a manufacturing system (see Example~\ref{ex1}). 

Using our approach, we also identified that the precise range propagation problem within Simulink models and even the deep neural networks (DNNs), which is key to static analysis but that previously often suffered from over-approximation under input dependency, as an instance of the strongest system-property computation. Our QE-based composition and range propagation approach has been adopted by an industrial avionics application V\&V tool suite, Honeywell Integrated Lifecycle Tools \& Environment (\texttt{HiLiTE}) \cite{bhatt2010towards, ren2016improving} for its state-of-the-art static analysis.

Another important innovation of the paper is the extension of QE-based composition approache to the time-dependent scenarios, where a property can be evaluated upon a finite history of variables. For example, in a manufacturing system with bounded-processing delay $d$ (so every arrival at time step $k$, departs from the system by the time step $k+d$), a cascade of two such systems will have a total delay of $2d$. In the composition of such time-dependent atomic-properties using QE, we show that the composed system-property will in general involve a constraint over a longer history of inputs and outputs, but within the cumulative history lengths of all its component atomic-properties. We introduce, accordingly, the notion of {\em property order}, a method to infer the corresponding {\em system order}, and the inference of the strongest system-property from the composition of atomic-properties and their time-shifted replicas.

We have developed our QE-based compositional verification prototype verifier, namely \texttt{ReLIC} (Reduced Logic Inference for Composition), based on the integration of an open source QE solver \texttt{Redlog} with a compositional analyze \texttt{AGREE} \cite{stewart2017architectural, gacek2015towards, cofer2012compositional}. Our implementation uses  \texttt{AGREE}'s front-end for the specification of system architecture/connectivity, components, and their properties in the modeling language AADL\cite{feiler2006architecture} and its \texttt{AGREE} annex, and interacts with \texttt{Redlog} at the back-end to output the composition result. This feature of our approach is also adopted within the industrial V\&V tool suite \texttt{HiLiTE} for augmented invariant derivation of Simulink models with feedback loops.

Additional to the aforementioned work, we also reduce the satisfiability problem often formulated from model-checking into an instance of QE. Specifically, we look into the $k$-induction \cite{sheeran2000checking, kahsai2011pkind, ghassabani2017efficient} in bounded model checking, implemented in \texttt{JKind} \cite{gacek2015jkind} to verify the bounded time invariant properties of systems specified in Lustre language \cite{halbwachs1991synchronous}. %Under this scheme, to prove a transition system satisfies some invariant property, one needs to prove the base case and the inductive case for some $k$. 
In original $k$-induction, both base and inductive verification steps are interpreted as instances of an SMT (Satisfiability Modulo Theory) \cite{barrett2009satisfiability} problem. Hence, QE can be used as an alternative for common SMT-solvers, since satisfiability checking of a formula %$\phi(x_1,\dots,x_n)$ in $n$-variables 
is equivalent to checking the truth of the corresponding existentially quantified formula. %$\exists x_1\dots\exists x_n \phi(x_1,\dots,x_n)$ evaluates to $true$ or $false$. %Thus with regards to satisfiability, the capability of SMT solvers and QE solvers overlap, and can vary depending on the algorithms they employ and the theories they support. 
\cite{jovanovic2013solving} has shown through the accumulated experimental data that QE solvers like \texttt{Redlog} in general are more powerful than SMT solvers like \texttt{Z3},  \texttt{cvc3},\texttt{iSAT}, in certain domain of non-linear arithmetic, regarding the execution time, and the range of problems those can solve. We have integrated \texttt{Redlog} with \texttt{JKind} for such enhancement on nonlinearity SMT-solving ability through quantifier-elimination. A related application of using QE to generate property-directed invariants in a $k$-induction based framework can be found in \cite{champion2015generating}. We summarize our key contributions presented in this paper in below:
\begin{itemize}
\item Establish quantifier elimination as a foundational calculus for property composition.
\item Develop a QE-based framework that performs composition of properties that depend on the input/output variables of either only the current step (termed time-independent) or over a history of past steps (termed time-dependent). The composition approach is uniform regardless of the interconnection types (cascade/parallel/feedback) for both time-dependent and time-independent properties. %, and also works for input-complete as well as input-incomplete properties.
\item Introduce the notion of the strongest system-property that can be inferred from the components' atomic-properties and their connectivity, using a QE-based composition approach.
\item Translate the verification of system-level postulated property from the derived strongest system-property into an instance of QE using universal quantifiers for automation of system-level property verification.
\item Add QE as a complementary option for SMT-solving for model-checkers.
\item Implement all the above in a prototype verifier named \texttt{ReLIC}, that employs AADL and \texttt{AGREE} Annex for component and property specification and \texttt{Redlog} for quantifier elimination to automatically compute the strongest system-property. Also integrate QE solver \texttt{Redlog} with the model-checker \texttt{JKind}.
\item Demonstrate our implementations on several illustrative and practical examples, and also
introduce the integration of our QE-based approaches into industrial V\&V tool suite \texttt{HiLiTE} for advanced static analysis of Simulink models and DNNs.
\end{itemize}

\subsection{Related work}
The first implementable quantifier elimination procedure, introduced by Collins in 1975, is based on cylindrical algebraic decomposition (CAD) \cite{collins1975quantifier} of double exponential complexity. Later on, Tarski showed that the first-order logic (FOL) over the real closed field admits QE \cite{tarski1998decision}. However, the original CAD algorithm for reals is of double exponential complexity. After decades of effort on improving QE techniques, tools capable of practical problem solving with various specialized procedures for restricted problem classes have been further developed and enriched, such as \texttt{Redlog}, \texttt{Qecad}, and \texttt{Mathematica}. For instance, \texttt{Redlog} \cite{redlog} implements algorithms based on \textit{virtual substitution} \cite{weispfenning1988complexity, weispfenning1990complexity} and \textit{partial CAD} \cite{collins1998partial}, that run fast for property expressions with low degree polynomial terms of quantified variables. 

Definition of interfaces with contracts and their composition is also studied in \cite{tripakis2011theory, benveniste2017synchronous}. In particular, \cite{tripakis2011theory} (that extends the concepts of \cite{benveniste2017synchronous}), introduced two separate definitions for cascade (termed connection in their paper) and feedback compositions (Definition 9 and 12, respectively). It turns out that in our approach, the same formula computes the composed system-level property regardless of whether the composition is cascade or feedback or anything else. Our proposed composition is (i) simple (involving quantifier elimination; see the formula in Theorem~\ref{thm:01} and compare it to those in Definitions 9 and 12 of \cite{tripakis2011theory}), (ii) universal (the same form of formula works for all sorts of interconnections), (iii) automatically internalizes the variables that are no longer available externally to the system, and so no extra step of ``hiding'' (viz., Definition 13 of \cite{tripakis2011theory}) is needed, (iv) proves that the composed property is the strongest possible, and (v) finally extends the approach also to time-dependent properties, by way of introducing the time-shifted replicas of component properties.

Hierarchical modeling are being adopted for complex and safety-critical applications \cite{lu2013hierarchy, duan2018refinement} for efficient incorporate of formal methods in the verification process. Many systems can express their properties or requirements formally in the first-order logic formulas. Hence QE has been widely adopted broadly. For example, in \cite{jirstrand1997nonlinear}, QE was used for nonlinear problem solving in continuous control system design. \cite{lafferrierre2001symbolic, anai2001reach} used QE for exact reachable set computation of linear dynamics with certain eigen-structures and semi-algebraic initial sets. This method was then generalized in \cite{tiwari2003approximate} on linear systems with a broader range of eigen-structures. Furthermore, \cite{taly2010switching, sturm2011verification} developed QE-based verification and synthesis techniques for continuous and switched dynamical systems. QE solvers are also used for the bounded model-checking (BMC) as back-end reasoning egines in \cite{kwiatkowska2015synthesising}, as well as in the theorem-prover \texttt{KeYmaera} \cite{platzer2008keymaera}  for hybrid systems.

%The rest of the paper is organized as follows. Section~\ref{sec:QEforV} describes the integration of \texttt{Redlog} with \texttt{JKind} to provide the model-checker an additional solver option. Section~\ref{sec:ReLIC} describes our framework for QE-based property composition, along with the prototype tool \texttt{ReLIC} for time-independent property composition. Section~\ref{sec:ReLICforTD} provides the extension of QE-based property composition to compose time-dependent properties. Each of these sections provides illustrative examples of the said implementations. Section~\ref{sec:applications} showcases extended analysis applications performed by our QE-based approaches integrated into an industrial V\&V tool suite. Finally, the conclusion and future work are discussed in Section~\ref{sec:conclusion}.

%\subsubsection{Subsubsection Heading Here}
%Subsubsection text here.

%\IEEEpubidadjcol
\section{QE support for Verification: Integration of \texttt{Redlog} with \texttt{JKind}}
\label{sec:QEforV}
\subsection{Preliminary}
$k$-induction based bounded model checking (BMC) has been widely used for the time-dependent system verification. Here we show how QE solver can serve as the back-end solver aiding $k$-induction proofs. We examine the $k$-induction algorithms adopted in the tool JKind \cite{cofer2012compositional,gacek2015jkind} for transition systems written in Lustre language \cite{halbwachs1991synchronous}. JKind is developed upon its precursor \texttt{Kind} \cite{kahsai2011pkind}, with the advance of  platform independence and easy integration interface in Java environment. 

We first briefly describe the base step and inductive step of $k$-induction proof \cite{sheeran2000checking, kahsai2011pkind} and their formulations. Given a transition system $S$ of state variable vectors $x$ in some formal logic language $\mathcal{L}$, specified with initial condition $I(x)$ and transition relation $T(x,x')$. Let $x(i)$ denote $x$ at the $i^\text{th}$ time step. A property of $S$ (i.e., a predicate over $x$) , denoted by $\phi(x)$, is a valid property of $S$, if it satisfies both the base condition and inductive condition in $\mathcal{L}$ for some finite $k\in\mathbb{Z}_{\ge1}$ and for any $n\in\mathbb{Z}_{\ge0}$:
\begin{itemize}
\item base condition: $I(x{(0)})\wedge \bigwedge_{i=0}^{k-2}T\big(x{(i)},x{(i+1)}\big)\Rightarrow \bigwedge_{i=0}^{k-1} \phi(x{(i)})$;
\item inductive condition: $\bigwedge_{i=n}^{n+k-1} T\big(x{(i)},x{(i+1)}\big)\wedge \bigwedge_{i=n}^{n+k-1} \phi(x{(i)})\Rightarrow \phi(x{(n+k)})$.
\end{itemize}
The base condition checks that $\phi(x)$ is satisfied for the beginning $k$ steps from the initial step.  The violation of the base condition suggests the existence of at least one concrete counterexample that invalidates $\phi$. The inductive condition checks that $\phi$ is satisfied at one step given only that all of its $k$-step precursor satisfy $phi$. A counterexample from the inductive condition check is not necessarily concrete since the $k$-step precursor may not be reachable from any initial state of $S$. On the other hand, when both base condition and inductive condition hold for some finite $k$, $\phi$ is guaranteed to hold over all steps of all reachable traces. A $k$-induction proof often increments $k$ by 1 from 0, excluding any spurious counterexamples from the induction condition check. Because $k$-induction process may never terminate ($k$ may goes up to infinity before conclusion) due to its inherent undecidable nature, a pre-determined bound on $k$ is specified as a hard stop.

The architecture of \texttt{JKind} is illustrated in Figure~\ref{fig:JKindArch}. Its input script is in Lustre language, specifying system description and properties to be verified. \texttt{JKind} initiates three processes for base condition check, inductive condition check, and invariant generation respectively. The invariant generation works in parallel withe the based and inductive cases independently, searching for invariant candidates from existing templates, that is possible to accelerate the induction condition check by strengthening its premise. Each process has its own copy of back-end SMT solver engine. A director program coordinates three processes asynchronously and outputs verification result in a user-friendly interface.

\vspace{-0.1in}
\begin{figure}[!htb]
\centering
\includegraphics[scale=0.41]{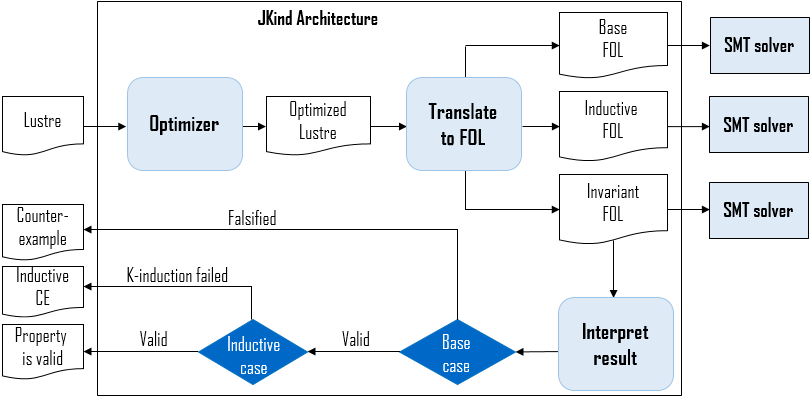}
\caption{\texttt{JKind} architecture.}
\label{fig:JKindArch}
\vspace{-0.1in}
\end{figure}

\vspace{-0.1in}
\subsection{Reduction of SMT instance to QE instance}
An SMT instance is a FOL formula, whose satisfiability needs to be checked. An unquantified FOL formula $\phi(x_1,\dots x_n)$ is satisfiable is defined as there \textit{exists} at least one assignment of the variables with which the formula evaluates to $true$. This existence of assignment can be expressed as an existential (as opposed to universal) QE instance, i.e., $\exists x_1\dots \exists x_n \phi(x_1,\dots x_n)$. Since the QE instance has now all its variables quantified in the formula, its equivalent quantifier-free formula is either $true$ or $false$. In the case of former, a satisfiable assignment of $(x_1\dots x_n)$ is also returned. In another word, check if  $\phi(x_1,\dots x_n)$ is satisfiable as an SMT problem is reduced to checking if $\exists x_1\dots \exists x_n \phi(x_1,\dots x_n)$ is equivalent to $true$ as a QE problem.

State-of-the-art SMT solvers such as \texttt{Z3} \cite{z3} are commonly used problem solving engines in the research community. \texttt{Z3} supports various data structures such as lists, arrays and bit vectors, and is capable of solving linear/non-linear problems in mixed boolean/integer/real domain. The QE solver \texttt{Redlog}, on the other hand, works exclusively with interpreted FOL formulas containing atoms from one particular \texttt{Redlog}-supported domain, which corresponds to a choice of admissible functions and relations with specified semantics, including for example non-linear real arithmetic (Tarski Algebra), parametric quantified Boolean formulae, Presburger arithmetic, and others.

We have implemented two new modules in Java for the QE augmentation of $JKind$. The first module \texttt{S2RTool} is to translate the SMT formulas redirected from \texttt{JKind} in SMT-Lib 2.0 format to \texttt{Redlog} input format, based on the parser automatically generated via \texttt{Antlr v4} \cite{Antlr}. The second module is an interpreter that processes the \texttt{Redlog} result back to \texttt{JKind} for termination output or the next induction proof iteration. Figure~\ref{fig:RedlogJKindIntegration} shows the integration of the two new modules with \texttt{JKind} and the data flow redirection. The dotted arrows show the new route of the QE capability, whereas the solid arrows show the original route of SMT capability. Figure~\ref{fig:Z3vsRedlog} shows a simple comparison between \texttt{Z3} and \texttt{Redlog} input scripts for the same underlying problem.

\begin{figure}[!htb]
\centering
\includegraphics[scale=0.38]{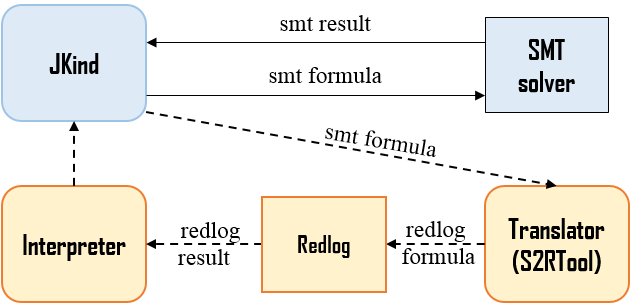}
\caption{Data redirection in \texttt{Redlog}-integrated \texttt{JKind}.}
\label{fig:RedlogJKindIntegration}
\end{figure}
\vspace{-0.15in}

\begin{figure}[!htb]
\centering
\includegraphics[scale=0.6]{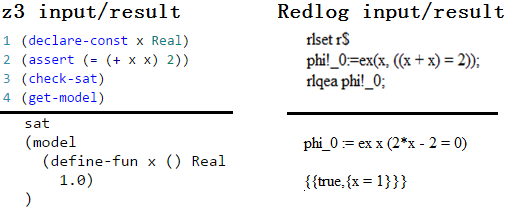}
\caption{\texttt{Z3} vs. \texttt{Redlog}.}
\label{fig:Z3vsRedlog}
\end{figure}

\begin{remark}\rm
\label{rmk:01}
\texttt{Redlog} setting does not allow problems in mixed real-integer domain, but a QE problem involving both real and integer type of variables can still be handled in unified real domain, by mapping integer variables to real variables with extra constraints as follows: each bounded integer variable $x$, is defined as real appending the auxiliary formula $\bigvee_{i=1}^n(x=v_i)$, where $\{v_1,\dots,v_n\}$ is the finite set of integer values within the original range of $x$. An infinitely-enumerated integer variable is simply defined as real, noting that the original problem thereby has been over-approximated. By the nature of over-approximation, a $false$ result of the over-approximated problem guarantees a $false$ result of the original problem. In contrast, a $true$ result has to be examined carefully. Only returned assignment of integer type is accepted by the original problem. Otherwise, we simply refine the over-approximated problem by adding the negation of the spurious assignment to exclude it from the result of the next iteration of QE solving. 
\end{remark}

\subsection{Experimental result}
We have successfully experimented our \texttt{Redlog}-integrated \texttt{JKind} on a set of Lustre-based programs including a nonlinear fuzzy logic model. The fuzzy logic model executes one out of 54 different nonlinear computation actions, depending on its decision making (linear) conditions evaluated by 4 real inputs $x_1,x_2,x_3$ and $x_4$. Each execution invokes a call to a $4^\text{th}$-order polynomial function to compute the value for the common output variable $y$. All the input/output variables are of real type while the execution switch variable $N$ is an integer that enumerates from 1 to 54. For instance, if the inputs satisfies the condition for $N = 1$, the output is computed by the $4^\text{th}$-order polynomial: $y=-2.22222x_1-2x_2-4x_3+10x_4+8.88889x_1x_2+7.40741x_1x_3+59.25926x_1x_4+12x_2x_3+32x_2x_4+40x_3x_4-59.25926x_1x_2x_3-177.77778x_1x_2x_4-74.07407x_1x_3x_4-240x_2x_3x_4+888.88889x_1x_2x_3x_4+1$. 

There are two properties to be verified, one of which is to check whether the absolute value of the output of the fuzzy logic is always bounded by 1. \texttt{JKind} is unable to resolve this since all of its back-end SMT solvers (\texttt{Z3}, etc.) fail to terminate on this particular nonlinear problem. On the contrary, our \texttt{Redlog}-integrated \texttt{JKind} concludes that one property is valid, and the other property is invalid, for which it reports a concrete counterexample.  Figure~\ref{fig:fuzzyModelResult} shows a brief log of QE-integrated \texttt{JKind} execution and results, in which the entire execution takes less than 16 seconds on a standard laptop. A detailed log (omitted here) shows that texttt{Redlog} responses to each QE query in less than 1 second.

\begin{figure}[!htb]
\centering
\includegraphics[scale=0.58]{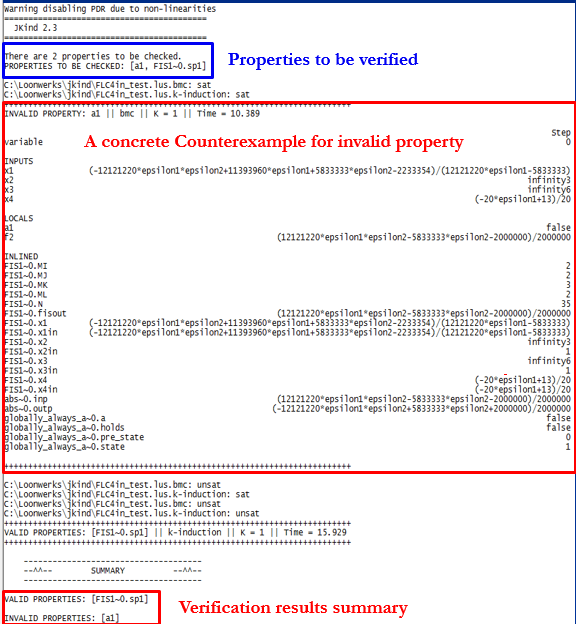}
\caption{Verifying a fuzzy logic model using QE-integrated \texttt{JKind}.}
\label{fig:fuzzyModelResult}
\end{figure}

\section{Time-independent property composition and \texttt{ReLIC} implementation}
\label{sec:ReLIC}
A modular approach to establish system correctness involves the compositional paradigm \cite{henzinger2001assume, cofer2012compositional, bogomolov2014assume, mehrpouyan2016complex, stewart2017architectural}, within which a property of a module (component/system) is specified by a contract of assume-guarantee (A-G) style. 
%represented by a pair $(A,G)$, where $A$ and $G$ are FOL formulae: $G$ describes the guaranteed behavior of the module (typically a predicate over the input and output variables) while $A$ describes the assumed behavior of the environment with which the module interacts (typically a predicate over the input variables). 
%In many developments, $(A,G)$ is treated as its logical equivalence $[A\Rightarrow G]\equiv (A\wedge G)\vee\neg A$ in First Order Logic format.
Another direction of compositional reasoning and verification aims to infer the system-level property from the atomic-properties of its components together with their interconnection. Existing work of such \cite{alur1999reactive, tripakis2011theory, broy2012specification} present such composition under certain relational interface framework, handling the composition of different interconnection types (cascade/parallel/feedback) differently.

\cite{cofer2012compositional} introduced a compositional approach for modular system, where verification a postulated contract property for a system with $N$ components from the given components' contracts is decomposed to $N+1$ verification steps: one for each component and an extra one for checking the system as a whole. Each component-level verification step establishes the guarantee of the component by showing that its assumption is met under the postulated system-level assumption and the guarantees of all its upstream components. After all the components are checked to behavior as expected in their contract guarantees, the final verification step establishes that the postulated system-level guarantee is met under the guarantees of all the components. Such $(N+1)$-step compositional reasoning is implemented in the tool \texttt{AGREE}. It describes the system architecture in AADL language\cite{feiler2006architecture}, and specifies the property contracts in the \texttt{AGREE} annex. \texttt{AGREE} is developed as a plug-in tool within the Eclipse-based OSATE2 \cite{OSATE} platform that provides AADL v2 grammar support, and it uses \texttt{JKind} as its back-end model-checker for each of the $N+1$ verification steps. Figure~\ref{fig:AGREEArch} shows the architecture of \texttt{AGREE} and its OSATE2 platform.

\begin{figure}[!htb]
\centering
\includegraphics[scale=0.6]{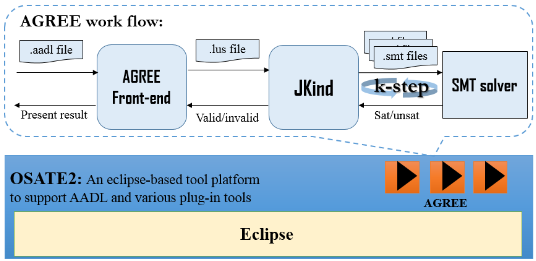}
\caption{Architecture of \texttt{AGREE} within OSATE2 platform.}
\label{fig:AGREEArch}\vspace{-0.1in}
\end{figure}

We propose a QE-based framework that introduces the notion/proof of the \textit{strongest system-property} derived from composition for time-independent as well as time-dependent properties in a single step, and also 
%for input-complete as well as input-incomplete cases.
performs property composition in a \textit{uniform} way regardless of the interconnection types (cascade/parallel/feedback). 

\subsection{Proposed QE-based compositional reasoning and verification}
%Our QE-based compositional reasoning and verification is based upon the notion of  \textit{system-property}, and further the inference of the \textit{strongest system-property}, given the components' \textit{atomic-properties} and system's interconnection relation. 
We start with the following definitions to formalize the concepts of atomic- versus system-property.

\begin{definition}\rm
\label{def:01}
An {\em atomic-property}, of a component, is a first-order predicate over sequences of component inputs and outputs; it is said to be {\em time-independent} if the predicate is over a single time-step of inputs and outputs, and otherwise it is called {\em time-dependent}.
\end{definition}

\begin{definition}\rm
\label{def:02}
A {\em system-property}, of a system of connected components, is also a first-order predicate over its sequences of system's inputs and outputs {\em that is implied} by the atomic-properties of the components and their interconnection relations (as captured as the equalities over the inputs and outputs of the components).
\end{definition}

\begin{definition}\rm
\label{def:03}
The {\em strongest system-property}, of a system of connected components, is a system-property that is stronger than any other system-property. (A predicate $p$ is said to be stronger than predict $q$ if and only if $p\Rightarrow q$ holds.)
\end{definition}

\begin{example}\rm
\label{ex1}
In a cascade connection of two identical manufacturing modules, suppose each satisfies a buffer overflow {\em atomic-property}: $\#_{in_1} - \#_{out_1} < B$ and $\#_{in_2} - \#_{out_2}< B$. Also, the cascade interconnection relation can be captured as $\#_{out_1}=\#_{in_2}$. Then for a cascade of two copies, one can verify that $\#_{in_1} - \#_{out_2}\leq 2B$ is a {\em system-property}, by showing  $\big((\#_{in_1} - \#_{out_1}< B)\wedge(\#_{in_2} - \#_{out_2}< B)\wedge (\#_{out_1}=\#_{in_2})\big)\Rightarrow(\#_{in_1} - \#_{out_2}\leq 2B)$. Its {\em strongest system-property} is given by, $\#_{in_1} - \#_{out_2}<2B$, the inference of which is discussed below. 
\end{example}

Consider a system $S$ composed of $N$ components. Let $X:=\{x_1,\dots,x_n\}$ be the set of variables in $S$, $X_\text{int}:=\{x_1,\dots,x_{m\le n}\}\subseteq X$, be the set of internal variables (namely, all the variables hidden from the system interface), $X_\text{sys}:=X\setminus X_\text{int}=\{x_{m+1},\dots,x_n\}$ be the set of external variables (namely, the inputs/outputs of $S$), and $C:=\{(x_p,x_q)~|~x_p$ and $x_q$ are variables of connected ports in $S\}$ be the interconnection relation set. Let $V$ denote the assignments of variables in $X$, and $V_\text{sys}$ denote the assignments of variables in $X_\text{sys}$. A \textit{valid} assignment of a predicate $\phi$ is an assignment such that $\phi$ evaluates to $true$. Suppose the $i^\text{th}$ component's atomic-property is described by a property $\phi_i$. A uniform way of automatically deriving the strongest system-property for time-independent properties is given next.

\begin{theorem}\rm
\label{thm:01}
The {\em strongest system-property}, established upon the given component properties (as time-dependent atomic-properties) and the interconnection relation of system $S$ (as described by $C:=\{(x_p,x_q)~|~x_p$ and $x_q$ are variables of connected ports in $S\}$) is given by,
%\vspace{-0.1in}
\begin{equation}
\label{eqn:01}
\exists x_1\dots\exists x_m\Big(\bigwedge_{i=1}^N \phi_i\wedge\bigwedge_{(x_p,x_q)\in C}(x_p=x_q)\Big).
\end{equation}
\end{theorem}

\begin{proof}
We first prove that \eqref{eqn:01} is a system-property, and then prove that it is stronger than any other system-property.

Firstly, \eqref{eqn:01} can be rewritten as $\exists x_1\dots\exists x_m(\Phi)$, where
\begin{equation*}
\Phi:=\bigwedge_{i=1}^N \phi_i\wedge\bigwedge_{(x_p,x_q)\in C}(x_p=x_q),
\end{equation*}
is the conjunction of the components' atomic-properties and the interconnection relation. It's easy to check that $\Phi\Rightarrow$~\eqref{eqn:01} always holds. Also note that \eqref{eqn:01} quantifies the internal variables in $X_{int}$, and as a result is a predicate over only system-level input/output variables. It follows that \eqref{eqn:01} is a system-property.

Secondly, by the semantic meaning of \eqref{eqn:01}, for any valid assignment $\vec v_{sys}=(v_{m+1}\dots v_n)\in V_\text{sys}$ of \eqref{eqn:01}, there exists a valid assignment $\vec v=(v_1\dots v_n)\in V$ of $\Phi(\vec v)$. Also for any system-property $\phi_{sys}$, we have $\Phi\Rightarrow\phi_{sys}$ by definition of system-property. So $\Phi(\vec v)=true$ in turn implies $\phi_{sys}(\vec v_{sys})=true$ (only the assignment $v_{m+1}\dots v_n$ is relevant to $\phi_{sys}$). Thus whenever \eqref{eqn:01} is $true$, $\phi_{sys}$ is also $true$, i.e., \eqref{eqn:01}$~\Rightarrow\phi_{sys}$. Thus, \eqref{eqn:01} is stronger than any other system-property, completing the proof. 
\end{proof}

\begin{example}\rm
Recall the cascade manufacturing system in Example~\ref{ex1}. Its strongest system-property can be derived from a QE process as below:
\begin{equation*}
\begin{split}
\exists\#_{out_1}\exists\#_{in_2}\big((\#_{in_1} - \#_{out_1}<B)\wedge(\#_{in_2} - \#_{out_2}<B)\\\wedge(\#_{out_1}\text{=}\#_{in_2})\big)\stackrel{\text{qe}}{\equiv}(\#_{in_1} - \#_{out_2}<2B).
\end{split}
\end{equation*}
It is strongest, in the sense that it implies any other system-property, as proved in Theorem~\ref{thm:01}. 
\end{example}

%\begin{remark}\rm
%Contracts composition is also presented in \cite{de2001interface, tripakis2011theory, benveniste2017synchronous} under certain definitions of relational interfaces. For example in \cite{tripakis2011theory}, the contract is syntactically defined as logical formalism which does not necessary take A-G style. In contrast,  \cite{benveniste2017synchronous} adopt the A-G style contract. Note a logical contract formula $C$ can be trivially expressed as a A-G contract $(true, C)$,  and conversely, an A-G contract $(A, G)$ is equivalent to a logical contract $\neg A \vee G$. Consider for example a component that performs division with inputs $x_1, x_2$, output $y$, and rejects a zero input as denominator. Then this contract can be written as, $C_{Div}:= x_2\neq 0\wedge y=x_1/x_2$, or equivalently in A-G style as, $(true, x_2\neq 0\wedge y=x_1/x_2)$.
%\end{remark}

%\begin{remark}\rm\label{rmk:rv1}
%In \cite{benveniste2007multiple}, a definition of elimination of a variable (or port) $p$ from a contract $(A,G)$ is given by $(\forall p A, \exists pG)$. Under A-G style contract, this is equivalent to $\exists p(A\Rightarrow G)$, which is the strongest system-property in the sense of~\ref{eqn:01}, where the equivalence is achieved from $\exists p(A\Rightarrow G)\equiv \exists p(\neg A\vee G)\equiv\exists p(\neg A) \vee \exists pG\equiv \neg(\forall pA)\vee \exists pG\equiv(\forall p A\Rightarrow \exists pG)$, in which the last equivalent expression corresponds to the contract $(\forall p A, \exists pG)$.
%\end{remark}

\begin{remark}\rm
Theorem~\ref{thm:01} shows that property composition for a modular system is essentially a QE problem. From this perspective, we introduce a unified QE-based compositional verification procedure containing only two verification steps. First, we derive the strongest system-property, denoted by $\phi_\text{ssp}$, according to \eqref{eqn:01} composing all the component contracts and their interconnection relation. After QE, $\phi_\text{ssp}$ is simplified to a quantifier-free formula with only the system-level input/output variables. Secondly, we check if the quantifier-free $\phi_\text{ssp}$ logically implies the postulated system property $\phi_\text{postl}$ that also contains only system-level input/output variables. The second step, although not part of Theorem~\ref{thm:01}, is yet another QE problem, that is to check if $\forall x_{m+1}\dots\forall x_n (\phi_\text{ssp}\Rightarrow \phi_\text{postl})$ reduces to $true$ or $false$. 
\end{remark}

\begin{remark}\rm
An inter-connected system is said to be well-posed if all its internal signals are uniquely defined given the external signals and the internal states \cite{vidyasagar1980well, doyle2013feedback}. We note that in the proposed framework of QE-based property composition, no such requirement of well-posedness is explicitly needed because the existential quantification over the internal variables already accounts for the existence of some values the internal signals can take, and so well-posedness is implicitly ensured through existential quantification.
\end{remark}

\subsection{\texttt{ReLIC} implementation and experimental results}
We have developed a prototype verifier \texttt{ReLIC} employing the aforementioned strategy, integrating \texttt{Redlog} (for QE-solving) with \texttt{AGREE} (for system and component specification). \texttt{ReLIC} abstracts the system architecture description in AADL and the contracts specification in \texttt{AGREE} annex, and formulates them  into a QE problem formatted for \texttt{Redlog} input. Figure~\ref{fig:AGREERedlogDF} shows interaction between the front-end of \texttt{AGREE} and the back-end solver \texttt{Redlog}. It also illustrates how this interaction completely bypasses \texttt{JKind}, replacing the aforementioned $N+1$ verification steps required under the original compositional verification approach of \texttt{AGREE} with our two steps  QE-based approach. Particularly in our approach, deriving the strongest system-property is unified and automatic, given the component contracts as atomic-properties. This feature is not seen from the current compositional reasoning tools.

\begin{figure}[!htb]
\centering
\includegraphics[scale=0.45]{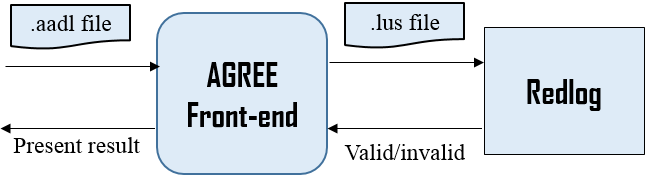}
\caption{Data flow between \texttt{AGREE} and \texttt{Redlog}.}
\label{fig:AGREERedlogDF}
\end{figure}

The illustrative model $S$ shown in Figure~\ref{fig:aModelExample} is taken from \cite{AGREEUserGuide}. It consists of three components, with their components and  interconnection relation architecture specified in AADL language. The properties for three components $A, B,$ and $C$, and the postulated system property system $S$ are as given by:
\begin{itemize}
\item $Property_A$: $In_A<20\Rightarrow Out_A<2\times In_A$;
\item $Property_B$: $In_B<20\Rightarrow Out_B<In_B +15$;
\item $Property_C$: $Out_C=In_C1+In_C2$;
\item $Property_{S}$: $In_S<10\Rightarrow Out_S<50$.
\end{itemize}

\begin{figure}[!htb]
\centering
\includegraphics[scale=0.3]{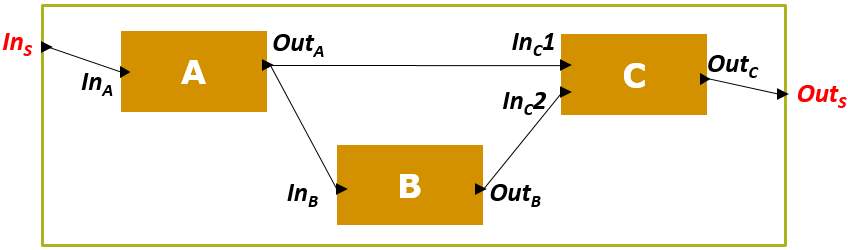}
\caption{Architecture of the example model modified from \cite{AGREEUserGuide}.}
\label{fig:aModelExample}%\vspace{-0.2in}
\end{figure}

\texttt{ReLIC} is executed by a new command ``Verify Composed Contract'' integrated into the \texttt{AGREE} front-end interface menu, as in Figure~\ref{fig:ReLICEnvironment}. The execution result of the above example in real domain is shown in Figure~\ref{fig:resultInReal}, while the derived strongest system-property $[In_S\le10\Rightarrow Out_S<4\times In_S+15]$ is shown in Figure~\ref{fig:sspInReal} displayed in the console window. In real domain, the postulated $Property_S$ is violated by, for example, the assignment value $(9,50)$ for $(In_S,Out_S)$. Figure~\ref{fig:ceInReal} shows one counterexample returned by \texttt{Redlog}. [Note the counterexample returned by \texttt{Redlog} may include constants denoted by indexed \textit{infinity} or \textit{epsilon}, where \textit{infinity} stands for positive and infinite value constant, and \textit{epsilon} stands for positive and infinitesimal value in the underlying field.]

In contrast, the postulated property is verified to be $true$ in the integer domain as shown in Figure~\ref{fig:resultInInt}. This is because the derived strongest system-property varies as the domain changes. Figure~\ref{fig:sspInInt} shows the derived strongest system-property in integer domain, in which the term $cong(p_1,p_2,p_3)$ represents a congruence with the non-parametric modulus given by the third argument. The strongest system-property of the integer domain, $[In_S\leq 10\Rightarrow Out_s\leq 4\times In_s+12]$ turns out to more strict (than the one in the real domain), resulting in the implication to the postulated $Property_S$. In both instances, the entire execution process terminates less than 1 second on a standard laptop.

\begin{figure}[!htb]
\centering
\includegraphics[scale=0.48]{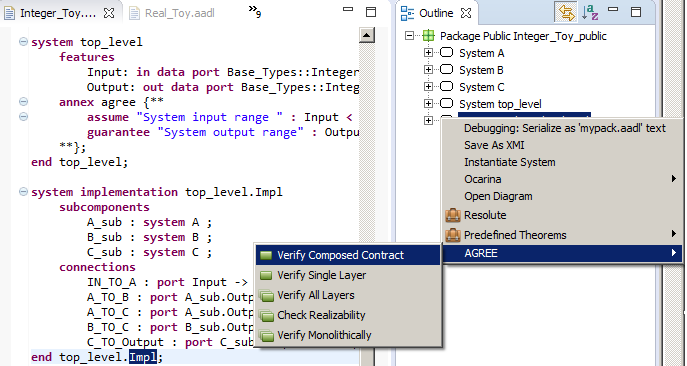}
\caption{Execution of \texttt{ReLIC} via \texttt{AGREE} command in OSATE2.}
\label{fig:ReLICEnvironment}
\end{figure}

\begin{figure}[!htb]
    \centering
    \begin{subfigure}[b]{0.48\textwidth}
            \centering
            \includegraphics[width=\textwidth]{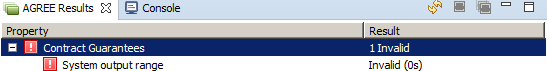}
            \caption{Verification result in real domain.}
	 \label{fig:resultInReal}
    \end{subfigure}
\begin{subfigure}[b]{0.48\textwidth}
            \centering
            \includegraphics[width=\textwidth]{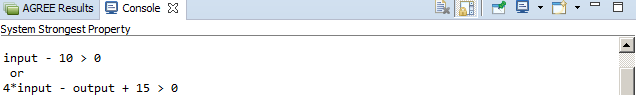}
	 \caption{The strongest system property derived in real domain.}
	 \label{fig:sspInReal}
    \end{subfigure}
\begin{subfigure}[b]{0.48\textwidth}
            \centering
            \includegraphics[width=\textwidth]{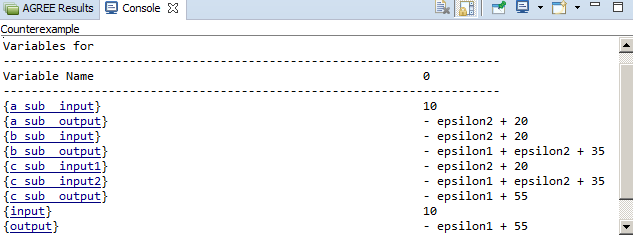}
	 \caption{Counterexample in real domain.}
	 \label{fig:ceInReal}
    \end{subfigure}
\begin{subfigure}[b]{0.48\textwidth}
            \centering
            \includegraphics[width=\textwidth]{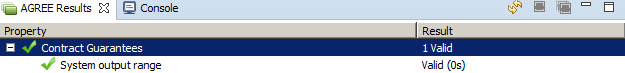}
	 \caption{Verification result in integer domain.}
	 \label{fig:resultInInt}
    \end{subfigure}
\begin{subfigure}[b]{0.48\textwidth}
            \centering
            \includegraphics[width=\textwidth]{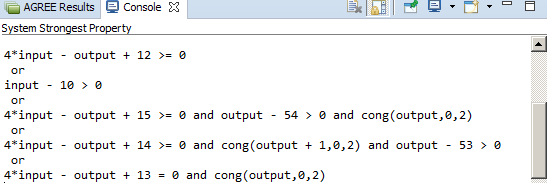}
	 \caption{The strongest system property derived in integer domain.}
	 \label{fig:sspInInt}
    \end{subfigure}
    \caption{\texttt{ReLIC} verification results on the illustrative example.}
\end{figure}

\section{Time-dependent property composition and \texttt{ReLIC} implementation}
\label{sec:ReLICforTD}
Complex systems often possess time-dependent features through components such as state-machine, counter, or PID controller. %While Definition~\ref{def:01}--\ref{def:03} apply also to time-dependent properties, a component atomic-property can be a predicate over its input/output variables at different time-steps. 
For performing property composition in the time-dependent, \eqref{eqn:01} cannot be used as is. This is observed in the example below.
\begin{example}\rm
\label{ex2}
Consider a {\em cascade} of two identical manufacturing components with bounded processing delay, where every input at time step $k$, departs from the component by time step $k+d$. For $i=1,2$, let $u_i{(k)},y_i(k)$ be the variables $u_i$ and $y_i$ at the $k^\text{th}$ time step. Then for the simple case of $d=1$, the two components can be formulated as $y_i(k)=u_i(k-1),i=1,2$ and their cascade connection by $y_1(k)=u_2(k)$. It is easy to see that the strongest system-property is ${y_2(k)=u_1(k-2)}$. Note this final formula includes the term $u_1(k-2)$ that was not provided by the given component property formulae, therefore the standard QE-based composition in \eqref{eqn:01} cannot be used directly to obtain ${y_2(k)=u_1(k-2)}$. Remarkably, by simply shifting each atomic-property by one time step, and composing these time-shifted replicas with the original atomic-properties  with all the internal variables $y_1(k)$, $y_1(k-1)$, and $y_2(k-2)$ existentially quantified: $\exists y_1(k)\exists y_1(k-1)\exists y_1(k-2)\big[\big(y_1(k)=u_1(k-1)\big)\wedge \big(y_2(k)=y_1(k-1)\big)\wedge \big(y_1(k-1)=u_1(k-2)\big)\wedge\big(y_2(k-1)=y_1(k-2)\big)\big]$, then upon QE, we do get the desired strongest system-property: $y_2(k)=u_1(k-2)$.
\end{example} 

\subsection{Approach to time-dependent property composition}
The example above shows that the time-dependent properties composition requires pre-processing regarding time-shift  before composition could happen even for the simplest cases. To be specific, the original atomic-properties may need to be time-shifted to match the (possibly) higher order of the system-property, where a system-property is a predicate over sequence of inputs and outputs. In time-dependent cases, component/system behavior at the current time step is governed by a set of predicates of its inputs and outputs over a finite history, which is captured by a set of difference equations/inequations over the component/system inputs and outputs. To formalize this time-shifting idea in the pre-processing, we introduce the notion of component/system order.
\begin{definition}\rm
For a component/system its {\em order} is defined as the difference of minimum and maximum time-shifts present in its governing difference equations/inequations predicates.
\end{definition}
\begin{remark}\rm
\label{rmk:3}
This definition of the order applies to both component and system. It is possible that the composed strongest system-property has a lower order than the component property. For instance, a system that has input $u$, internal variable $x$, and output $y$, possesses two atomic-properties $x(k)=-x(k-1)+u(k)$ and $y(k)=x(k-1)+x(k)$, each of which is of order 1 by our definition. But the strongest system-property $y(k)=u(k)$, obtained from variable substitution $y(k)=x(k-1)+x(k)=x(k-1)+\big(-x(k-1)+u(k)\big)=u(k)$, is of 0 order.
\end{remark}

This section begins by showing an example that adding time-shifted replicas of component atomic-properties before composition to yield the strongest system-property can be straightforward. One has to decide how many time-shifted replicas of each original atomic-property, and how many time-step to be shifted for each replica is sufficient and necessary. The key factor is the order $M_\text{sys}$ of the composed system. The following theorem presents a simple way to find an upper bound for the system order by summing up all its component orders.

\begin{theorem}\rm
\label{thm:02}
Given a system $S$ composed of $N$ multi-input-multi-output components, if the $i^\text{th}$ component is of order $M_i$, then the system order $M_\text{sys}$ of $S$ has an upper bound $\sum_{i=1}^N M_i$. 
\end{theorem}

\begin{proof}
Without loss of generality, choose any two connected components of $S$ with orders $M_1$ and $M_2$ respectively. Each is given a set of atomic-properties in the general form of nonlinear difference equations/inequations $f(\cdot)\sim 0$, where $\sim\in\{=,\ge,>\}$. For each inequation, a slack variable $u_f$ is introduced so that the {\em inequation} is rewritten to an equivalent form as $f(\cdot)-u_f=0\wedge u_f\sim0$. This additional inequation $u_f\sim0$ is of 0 order therefore does not contribution to the system order. 

For component $i_{i\in\{1,2\}}$ of the two chosen components, the set of its difference equations can be translated to an equivalent state-space representation \cite{fadali2012digital}, in which the number of the state variables equals its component order. The general form of the state-space representation for component $i$ is:
\begin{align*}
&{\bf x_i}(k+1)={\bf f_i}\big({\bf x_i}(k),{\bf u_i}(k)\big),\\
&{\bf y_i}(k)={\bf g_i}\big({\bf x_i}(k),{\bf u_i}(k)\big),
\end{align*}
where vectors ${\bf u_i}$ (size: $N_i\times1$), ${\bf x_i}$  (size: $M_i\times1$), and ${\bf y_i}$  $(size: O_i\times1)$ are the input (including the slack variables), state, and output variable vectors respectively. (Note that all the added constraints over the slack variables are of order 0 thus do not alter the component order.)  ${\bf f_i}(\cdot)$ (size: $M_i\times1$) and ${\bf g_i}(\cdot)$ (size: $O_i\times 1$) are vectors of functions. 

One can simply stack the two state space representations into a single one:
\begin{align*}
&\begin{bmatrix}
    {\bf x_1}(k+1)\\
    {\bf x_2}(k+1)
 \end{bmatrix}=
\begin{bmatrix}
    {\bf f_1}\big({\bf x_1}(k),{\bf u_1}(k)\big) \\
    {\bf f_2}\big({\bf x_2}(k),{\bf u_2}(k)\big)
  \end{bmatrix},\\
&\begin{bmatrix}
    {\bf y_1}(k)\\
    {\bf y_2}(k)
 \end{bmatrix}=
\begin{bmatrix}
    {\bf g_1}\big({\bf x_1}(k),{\bf u_1}(k)\big) \\
    {\bf g_2}\big({\bf x_2}(k),{\bf u_2}(k)\big)
  \end{bmatrix}.
\end{align*}
We show that when the two components get connected, the number of total states does not increase, and this is the basis that theorem~\ref{thm:02} is established upon. The connection of two components can be formulated as ${\bf u_c}={\bf y_c}$ where ${\bf u_c}$ (resp. ${\bf y_c}$) is an input (resp. output) vector from the union of the two components' inputs (resp. outputs). Note that the proper connection allows one input to connect with no more than one output but one output to connect with arbitrary number of inputs. The system of more than two components can be obtained by iteratively connecting one component to the existing composition at a time. It is sufficient to show that such iteration does not introduce any additional state variable. For notation compactness, we adopt the following state-space representation for the two-component composition:
\begin{align*}
&{\bf x}(k+1)={\bf f}\big({\bf x}(k),{\bf u}(k)\big),\\
&{\bf y}(k)={\bf g}\big({\bf x}(k),{\bf u}(k)\big),
\end{align*}
which possesses $M_1+M_2$ state variables. Let $(u_p, y_q)$ be a single connection of the connected system, where $u_p\in {\bf u}$, $y_q\in {\bf y}$, and $u_p$, $y_q$ are from different components. Let ${\bf u'}$ (resp. ${\bf y'}$) be the vector after removing $u_p$ (resp. $y_q$) from ${\bf u}$ (resp. ${\bf y}$). We then have $u_p(k)=y_q(k)=g_q\big({\bf x}(k),{\bf u}(k)\big) =g_q\big({\bf x}(k),{\bf u'}(k)\big)$ with $g_q\in {\bf g}$, as a result, the state-space representation is updated to:
\begin{align*}
&{\bf x}(k+1)={\bf f}\Big({\bf x}(k), {\bf u'}(k), g_q\big({\bf x}(k),{\bf u'}(k)\big)\Big),\\
&{\bf y}(k)={\bf g}\Big({\bf x}(k),{\bf u'}(k), g_q\big({\bf x}(k),{\bf u'}(k)\big)\Big).
\end{align*}
It is easy to see that the new form possesses input variables from ${\bf u'}$ ($u_p$ is internalized), output variables from ${\bf y}$ (or ${\bf y}\setminus\{y_q\}$ if $y_q$ also is internalized), while the state set remains ${\bf x}$. (A lower order state-space representation may be possible after further simplification.) As a result, the connected system has an upper bound $M_1+M_2$ on its overall order. The theorem thus follows by iteratively connecting the rest components to the composition.
\end{proof}

In many cases, a component is specified with not one but a {\em set} of properties (see for instance in Figure~\ref{fig:vehicleMdl} the component CNTRL of the vehicle benchmark). Such a component shall be treated as a set of single-property sub-components in our composition analysis, where their connections are implied by the common variables. Thus we introduce the following corollary:
\begin{corollary}\rm
\label{crlr:01}
The order of a component is bounded by the sum of all its atomic-property orders. The system order stated in Theorem~\ref{thm:02} is bounded by the sum of all the atomic-property orders of all its components. 
\end{corollary}

Consider a system $S$ with variables $X:=\{x_1,\dots,x_n\}$, in which $X_\text{int}:=\{x_1,\dots,x_{m|m\le n}\}\subseteq X$ is the set of internal variables, and $C:=\{(x_p,x_q)~|~x_p$ and $x_q$ are variables of connected ports$\}$ be the set of interconnection relation. Let $x{(k)}$ denote the variable at the $k^\text{th}$ time step, especially with $x{(0)}$ being at the initial step. Let $X{([s,t])}:=\{x{(k)}\vert x\in X, k\in[s,t]\}$ be the variables over the time interval $[s,t]_{s,t\in \mathbb{Z}_{\ge0}, s\leq t}$. Given the composed system order $M_\text{sys}$, the $i^\text{th}$ atomic-property of the components can be replicated $M_\text{sys}-M_{i}$ times, where $M_{i}$ is the order of the $i^\text{th}$ atomic-property, and the $j^\text{th}$ replica shifted $j$ time-steps $(j =1,\dots,M_\text{sys}-M_{i})$ to obtain the set of constraints over the variables $X{([k,k+M_\text{sys}])}$ of the system. Then we give the following theorem:

\begin{theorem}\rm
\label{thm:03}
The original and time-shifted replicas of atomic-properties of all the components over $X{([k,k\text{+}M_\text{sys}])}$, together with the original and time-shifted replicas of the interconnection relation, can be composed to infer the strongest system-property of their system $S$ as follows:
\begin{multline}
\label{eqn:02}
\Exists X_\text{int}{([k,k\text{+}M_\text{sys}])}\Big(\\\bigwedge \text{all component predicates over }X{([k, k\text{+}M_\text{sys}])}\\\bigwedge_{(x_p,x_q)\in C,j\in[k,k\text{+}M_{sys}]}(x_p(j)=x_q(j))\Big),
\end{multline}

Similarly %for a system of order $M_\text{sys}$, initial conditions over the first $M_\text{sys}$ steps ($0,\dots, M_\text{sys}\text{-}1$) suffice. So 
the system-initial-condition is obtained using:
\begin{multline}
\label{eqn:03}
\Exists X_\text{int}{([0,M_\text{sys}\text{-}1])}\Big(\\\bigwedge \text{all component predicates over } X{([0,M_\text{sys}\text{-}1])}\\\bigwedge_{(x_p,x_q)\in C,j\in[0,M_{sys}\text{-}1]}(x_p(j)=x_q(j))\Big).
\end{multline}
\end{theorem}
\begin{proof}
We only show the establishment of the strongest system-property; the establishment of the system-initial-condition is similar. Since the system order is $M_\text{sys}$, any system-property of order more than $M_\text{sys}$ can be reduced to an equivalent system-property of order equal to $M_\text{sys}$. The proof then follows by simply replacing all single-step variables in the proof of Theorem~\ref{thm:01} with their corresponding $M_\text{sys}$-step versions.
\end{proof}

Once the strongest system-property and system-initial-condition are encoded using \eqref{eqn:02} and \eqref{eqn:03}, verification of a postulated system-level property can be done using \textit{induction-based proof}. 

\begin{example}\rm
Consider three component atomic-properties $z(k)=y(k)+x(k)$, $y(k)=u(k-1)$, and $x(k)=w(k-1)$, where $x$ and $y$ are internal variables to be eliminated. According to Theorem~\ref{thm:02}, the system order has an upper bounded as the summation of the three components' orders: $0+1+1 = 2$. %On the other hand, it is easy to see through substitution that the strongest system-property over $z$, $u$, and $w$ is $z(k)=u(k-1)+w(k-1)$, which is only of order 1. 
The inference of the strongest system-property requires first shifting each atomic-property by the difference between its order and the system order upper bound ($2,1,1$ resp.), and then combining them using \eqref{eqn:02}:
\begin{multline*}
\exists x{(k-1)}\exists x{(k)}\exists x{(k+1)}\exists y{(k-1)}\exists y{(k)}\exists y{(k+1)} \\\Big(\big(z{(k-1)}=y{(k-1)}+x{(k-1)}\big)\wedge \big(z{(k)}=y{(k)}+x{(k)}\big)\\\wedge\big(z{(k+1)}=y{(k+1)}+x{(k+1)}\big)\wedge \big(y{(k)}=u{(k-1)}\big)\wedge\\\big(y{(k+1)}=u{(k)}\big)\wedge \big(x{(k)}=w{(k-1)}\big)\wedge \big(x{(k+1)}=w{(k)}\big)\Big).
\end{multline*}
Upon QE we obtain:
\begin{equation*}
\big(z{(k)}=u{(k-1)}+w{(k-1)}\big)\wedge\big(z{(k+1)}=u{(k)}+w{(k)}\big).
\end{equation*}

\end{example}

Note in the above example, $z{(k+1)}=u{(k)}+w{(k)}$ is a replica of $z{(k)}=u{(k-1)}+w{(k-1)}$ with one step-time shift, and hence redundant, meaning that the property can be simplified to the latter. This is formalized in the following remark.
\begin{remark}\rm
In general, when the system order $M_\text{sys}$ is less than the upper bound $\sum_{i=1}^N M_i$ given in Theorem~\ref{thm:02}, the QE result of \eqref{eqn:02} will contain $\sum_{i=1}^N M_i-M_\text{sys}$ copies of redundant expressions in conjunction. These redundant expressions are replicas of some simpler expression with time-shifts but provide no extra information about the system. Therefore, often we could eliminate them by identifying them in the conjunctive normal form (CNF) of the QE result.
\end{remark}

\begin{example}\rm
Recall Example~\ref{ex2} with the component atomic-properties $y_1(k)= u_1(k-1)$ for $k\ge 1$, and $y_2(k)=y_1(k-1)$ for $k\ge1$, augmented with atomic-initial-conditions $y_2{(0)}=0$ and  $y_1{(0)}=1$. Then in this case, \eqref{eqn:03} is encoded as:
\begin{multline*}
\exists y_1(0)\exists y_1(1)\Big(\big(y_2{(0)}=0\big)\wedge\big(y_2{(1)}=y_1{(0)}\big)\wedge\big(y_1{(0)}=1\big)\\\wedge\big(y_1{(1)}=u_1{(0)}\big)\Big),
\end{multline*}
which is equivalent to $\big(y_2{(0)}=0\big)\wedge\big(y_2{(1)}=1\big)$ through QE.
\end{example}

%\begin{remark}
 %Whether or not an equivalent quantifier-free version can be inferred depends on the generality of the logic. For example, for the case of Propositional Temporal Logic (PTL) \citep{gabbay1980temporal}, that is subsumed by FOL (first order logic), it is known that FOL over reals admits quantifier elimination \citep{tarski1998decision}. On the other hand, FOLTL is an instance of second-order logic (SOL), which does not admit quantifier-elimination (projection of a Borel set, that can be defined in the SOL of reals, is not a Borel set in general \citep{cohn2013measure}). 

%\end{remark}

\subsection{\texttt{ReLIC} implementation and experimental results}
Algorithm~\ref{alg:01} presents the overall algorithm of the time-dependent property composition and postulated property verification. $\Phi_{\text{all}}$ in Line 1 is the set of all constraints over $X$ obtained by the union of component atomic-properties and the interconnection relation. In line 2, the system order $M_\text{sys}$ is assigned with its upper bound according to Theorem~\ref{thm:02} and Corollary~\ref{crlr:01}. Line 3 collects all the constraints over $X{([k,k+M_\text{sys}])}$, denoted $\Phi_\text{all}{([k,k+M_\text{sys}])}$, comprising the set of all the needed time-shifted replicas of atomic-properties in $\Phi_\text{all}$. Line 4 formulates the property composition as $f_\text{ssp}$ over $X{([k,k+M_\text{sys}])}$, by internalizing $X_\text{int}{([k,k+M_\text{sys}])}$ in $\Phi_\text{all}{([k,k+M_\text{sys}])}$ as in \eqref{eqn:02}. \texttt{Redlog} then performs QE on $f_\text{ssp}$ to obtain the strongest system-property $\phi_\text{ssp}$ in line 5. The system-initial-condition $I_\text{sys}$ is initialized to the default value $true$ in line 6. If the component atomic-initial-condition set is non-empty, the quantifier-free equivalence of the  system-initial-condition $I_\text{sys}$ is derived based on \eqref{eqn:03} in lines 8-10. Line 12 model checks whether the strongest system-property $\phi_\text{ssp}$ implies the given postulated property $\phi_\text{postl}$ using $k$-induction by QE-integrated \texttt{JKind}, the result, $result$, of which could be $false$, $true$, or $unknown$. If $result$ is $false$, a counterexample $ce$ will be generated. This inductive proof is carried out entirely at the system-level, without having to go back to the component-level, as that is the desired intent of a compositional reasoning and verification approach.

\begin{algorithm}[tbh]
\SetAlgoNoLine
\KwIn{System $S$, set of all variables $X:=\{x_1,\dots,x_n\}$, set of internal variables $X_\text{int}:=\{x_1,\dots,x_{m\le n}\}$,  set of interconnection relations $C$, set of component atomic-properties $\Phi_\text{comp}$, set of atomic-initial-conditions $I_\text{comp}$, postulated property $\phi_\text{postl}$ over system input/output variables $X\setminus X_\text{int}$.}
$\Phi_\text{all}\gets \Phi_\text{comp}\cup\{(x_p^k=x_q^k)\vert (x_p,x_q)\in C\}$\;
$M_\text{sys}\gets \Sigma_{\phi\in\Phi_\text{comp}}$ Order$(\phi)$\;
$\Phi_\text{all}{([k,k+M_\text{sys}])}\gets \{\text{Shift}(\phi, i)\vert \phi\in\Phi_\text{all}, i\in[0,M_\text{sys}-\text{Order}(\phi)]\}$\;
$f_\text{ssp}\gets\text{Compose}\big(X{([k,k+M_\text{sys}])}$, $X_\text{int}{([k,k+M_\text{sys}])},\Phi_\text{all}{([k,k+M_\text{sys}])}\big)$\;
$\phi_\text{ssp}\gets$ Redlog$(f_\text{ssp})$\;
$I_\text{sys}\gets true$\;
\If {$I_\text{comp}\neq \emptyset$}
{
  $I_\text{all}{([0,M_\text{sys}-1])}\gets I_\text{comp}\cup \Phi_\text{all}{([0,M_\text{sys}-1])}$\;
  $f_\text{sys}\gets\text{Compose}\big(X{([0,M_\text{sys}-1])}$, $X_\text{int}{([0,M_\text{sys}-1])},I_\text{all}{([0,M_\text{sys}-1])}\big)$\;
  $I_\text{sys}\gets$ Redlog$(f_\text{sys})$\;
}
$(result,ce)\gets$ $k$-induction$(\phi_\text{ssp},I_\text{sys},\phi_\text{postl},K)$\;
\KwOut{$result,ce$}
\caption{QE-based compositional reasoning and verification for time-dependent properties.}
\label{alg:01}
\end{algorithm}

We enhanced our prototype verifier \texttt{ReLIC} for the time-dependent property composition of finite-order. Figure~\ref{fig:ReLICArch} shows the architecture of enhanced \texttt{ReLIC}. It use the front-end infrastructure of the \texttt{AGREE} tool for editing and parsing the AADL models and its output console for result display. The ``Property Composer'' module derives the strongest system-property based on Theorem~\ref{thm:02}, Corollary~\ref{crlr:01}, and Theorem~\ref{thm:03}. Next, the ``Induction Verifier'' module uses the derived strongest system-property to perform inductive proof of a postulated property. Both the composer and verifier modules use \texttt{Redlog} as the back-end solver.
\begin{figure}[!htb]
\centering
\includegraphics[scale=0.32]{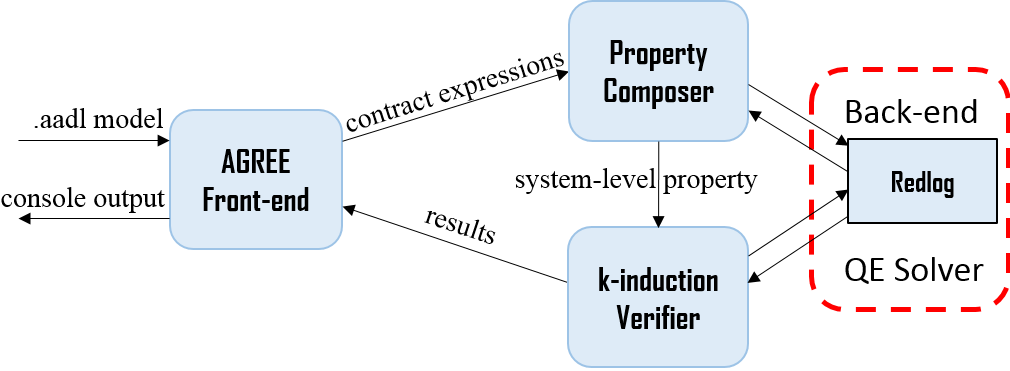}
\caption{ReLIC's Architecture}
\label{fig:ReLICArch}
\end{figure}

We tested our prototype verifier on a vehicle benchmark shown in Figure~\ref{fig:vehicleDynamics}. The benchmark model has a PID control ``Speed\_Control'' (CNTRL) component and a vehicle ``Throttle'' (THROT) component, interconnected in a feedback loop. The ``Speed\_Control'' component's dynamics is specified by a set of difference equations over its input and output as well as certain state variables (Figure~\ref{fig:CNTRLSpec}), whereas the ``Throttle'' component's dynamics is specified by a difference equation over only its input and output variables (Figure~\ref{fig:THROTSpec}). In \texttt{AGREE} annex notation, ${\bf prev}(\cdot)$ is the delay operator with the first argument denoting the delayed variable and the second argument denoting its initial value. For example, $u=0.2*e+0.2*{\bf prev}(e, 0.0)$ in Figure~\ref{fig:CNTRLSpec} can be rewritten as the equivalent form $(u(k) = 0.2*e(k) + 0.2*e(k-1)\big) \wedge \big(e(0)=0.0\big)$. For clarity and consistency, we use the later notation in the following illustration.

%\vspace{-0.1in}
\begin{figure}[!htb]
    \centering
    \begin{subfigure}[b]{0.49\textwidth}
            \centering
            \includegraphics[width=\textwidth]{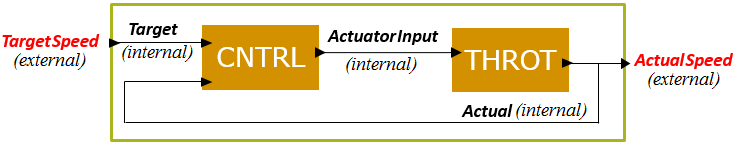}
            \caption{A vehicle model and its components dynamics.}
	 \label{fig:vehicleDynamics}
    \end{subfigure}
\begin{subfigure}[b]{0.45\textwidth}
            \centering
            \includegraphics[width=\textwidth]{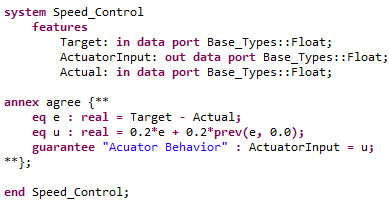}
	 \caption{Specification of CNTRL component.}
	 \label{fig:CNTRLSpec}
    \end{subfigure}
\begin{subfigure}[b]{0.45\textwidth}
            \centering
            \includegraphics[width=\textwidth]{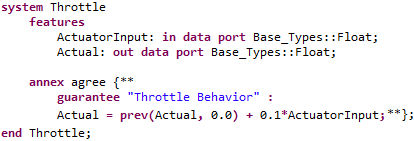}
	 \caption{Specification of THROT component.}
	 \label{fig:THROTSpec}
    \end{subfigure}
    \caption{A vehicle model, modified from \cite{AGREEUserGuide}.}
    \label{fig:vehicleMdl}
\end{figure}

There are two FOL expressions in total, namely, 
\begin{equation*}
u(k)=0.2*e(k) + 0.2*e(k-1),\text{ and }
\end{equation*}
\begin{equation*}
Actual(k)=Actual(k-1)+0.1*ActuatorInput(k),
\end{equation*}
The system-level postulated property is: 
\begin{equation*}
constTargetSpeed \Rightarrow (ActualSpeed < 1.0),
\end{equation*}
where $constTargetSpeed$ is a Boolean variable whose value indicates if the system target speed is set to constant 1.0 (see its defining expression in Figure~\ref{fig:ReLICOutput}). The \texttt{ReLIC} derived strongest system-property (based on \eqref{eqn:02} and \eqref{eqn:03} in Theorem~\ref{thm:03}) contains the system-level difference equation as well as the system-initial-condition over system input $TargetSpeed$ and system output $ActualSpeed$. The naive upper bound of the composed system order given by Theorem~\ref{thm:02} is $3$. However, the strongest system-property order is only 1 by observation, due to the inherent components' parallelism, as also computed by our approach (see the encoded expressions and the console output in Figure~\ref{fig:ReLICOutput}).

\begin{figure}[!htb]
\centering
\includegraphics[scale=0.56]{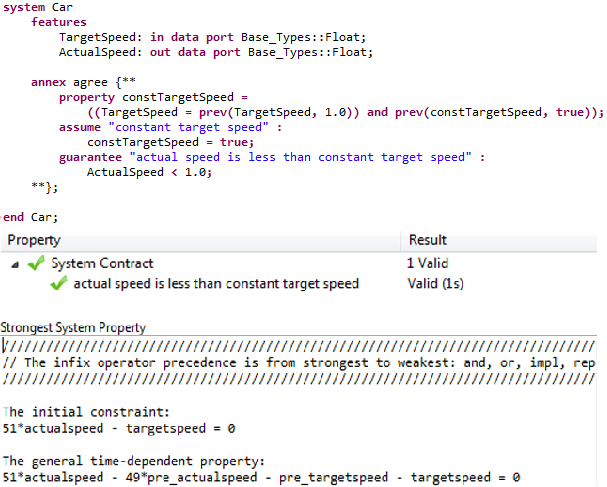}
\caption{\texttt{ReLIC} verification output (case-insensitive) on the vehicle model.}
\label{fig:ReLICOutput}
\end{figure}

\begin{multline*}
51*ActualSpeed(k)-49*ActualSpeed(k-1)\\-TargetSpeed(k-1)-TargetSpeed(k)=0,
\end{multline*} with the initial condition:
\begin{equation*}
51*ActualSpeed(0)-TargetSpeed(0)=0
\end{equation*}
We then use inductive proof with $k=1$ to verify the postulated property, where the base step is:
\begin{multline*}
(51*ActualSpeed(0)-TargetSpeed(0)=0)\Rightarrow\\\big((TargetSpeed(0)=1) \Rightarrow (ActualSpeed(0)<1)\big),
\end{multline*}
and the inductive step is
\begin{multline*}
\big((51*ActualSpeed(k)\text{-}49*ActualSpeed(k\text{-}1)\text{-}\\TargetSpeed(k\text{-}1)\text{-}TargetSpeed(k)=0)\wedge\\(TargetSpeed(k-1)=1) \wedge(ActualSpeed(k-1)<1)\big)\\\Rightarrow\big((TargetSpeed(k)=TargetSpeed(k-1))\Rightarrow\\(ActualSpeed(k)<1)\big).\end{multline*}

\section{Example applications of our QE-based strongest system-property computation approach}
\label{sec:applications}
At Honeywell Aerospace, scientists/engineers have developed the Honeywell Integrated Lifecycle Tools \& Environment (\texttt{HiLiTE}) \cite{bhatt2010towards, ren2016improving} tool suite for the automated verification of large and complex avionics systems developed using MATLAB Simulink/Stateflow. ``Range propagation'' is among its core static analysis functionalities that enables higher V\&V goals, by taking model input ranges and propagating those through each individual block to the model outputs. The propagated internal/output ranges can be used to identify potential model defects (e.g., checking if the range of the denominator port of a division block contains 0) or to guide automatic test generation. One of the common challenges in range propagation is the so-called \textit{dependency problem} \cite{moore2009introduction}, where if multiple inputs of a block are correlated (directly or indirectly) then the traditional output range calculation, that fails to capture the input dependency, may introduce undue over-approximation.

\begin{example}\rm
\label{ex4}
Figure~\ref{fig:AbsModel} shows a Simulink model implementation of the absolute value function using atomic blocks. Given input range $[-5, 5]$, its easy to see that the three inputs of the ``switch'' block range over $[-5,5]$, $[-5, 5]$, and $\{false, true\}$ respectively, after propagating through ``gain'' and ``lessEq'' blocks. Without a way of capturing the upstream structure that all three inputs of the switch block trace back to the same system input, a na\"{\i}ve way of range propagation through ``switch'' block will result in a conservative output range $[-5,5]$. 
\end{example}

\begin{figure}[!htb]
\centering
\includegraphics[scale=0.55]{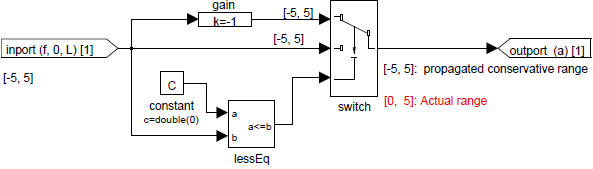}
\caption{Simulink model for absolute value function using atomic blocks.}
\label{fig:AbsModel}
\end{figure}

Using our QE-based notion, the strongest system-property provides the system-level input-to-output relationship as follows:
\begin{multline}
\label{eqn:04}
relation_{(out_{inport}, out_{switch})}:= \\\exists out_{gain} \exists out_{lessEq}\big(f_{gain}\wedge f_{lessEq}\wedge f_{switch}\big),    
\end{multline}
where $f_{gain}:=out_{gain}=-1\times out_{inport}$, $f_{lessEq}:=out_{lessEq}=0\le out_{inport}$, and $f_{switch}:= (out_{lessEq}\Rightarrow out_{switch}=out_{inport}) \wedge (\neg out_{lessEq}\Rightarrow out_{switch}=out_{gain})$ are the atomic-property predicates of the components ``gain', ``lessEq'', and ``switch'' respectively. Further, the precise range of $out_{switch}$ can be inferred using another QE-based strongest system-property formula:
\begin{multline}
\label{eqn:05}
range_{out_{switch}}:= \\\exists out_{inport}\big(range_{out_{inport}} \wedge relation_{(out_{inport}, out_{switch})}\big),
\end{multline}
where $range_{out_{inport}}:= out_{inport}\ge\text{-}5\wedge out_{inport}\le5$. Formula \eqref{eqn:05} simplifies to $range_{out_{switch}}:= out_{switch}\ge 0\wedge out_{inport}\le5$. Both~\eqref{eqn:04} and~\eqref{eqn:05} are instances of the strongest system-property formulation~\eqref{eqn:01}. Therefore, the derived range $[0,5]$ is the precise range that is stronger than any other system-property regarding the range. The same QE and strongest system-property based automated range propagation has been incorporated on \texttt{HiLiTE} for arbitrary components and interconnections.

The strongest system-property can also be used to derive the transfer function of a feedback structure within a Simulink model so that it can be then treated as a black box possessing that transfer function for simplification of further analysis.
\begin{example}\rm
\label{ex5}
Figure~\ref{fig:LagFilter} shows a Simulink model implementation of a lag filter function using atomic blocks. Using time-dependent version of the strongest system-property~\eqref{eqn:02} with system order equal to 1, the feedback structure within the enclosed box can be simplified into the following property:
\begin{multline}
\label{eqn:06}
 (40+tau)*out_{outport}(k+1)-  (40-tau)*out_{outport}(k)\\-out_{inport}(k) - out_{inport}(k+1)=0
\end{multline}
for $k\ge0$, along with the system-initial-condition $(40+tau)*out_{outport}(0) - out_{inport}(0)=0$ derived using~\eqref{eqn:03}.

It's easy to see that~\eqref{eqn:06} is the difference equation of a lag filter with the transfer function $\frac{1+z^{-1}}{(1+2\tau/T_s)+(1-2\tau/T_s)z^{-1}}$  in discrete-time domain, where $\tau=tau$ and $T_s=1/20$.
\end{example}

\begin{figure}[!htb]
\centering
\includegraphics[scale=0.58]{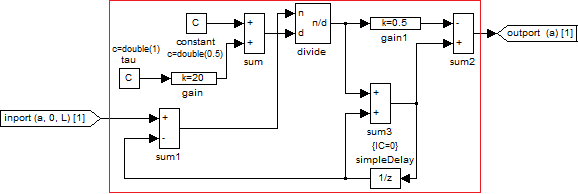}
\caption{Simulink model for a lag filter function using primitive blocks.}
\label{fig:LagFilter}
\end{figure}

A standard neural network of feed-forward structure is a layered system of interconnected nodes called \textit{neurons}, where each neuron consumes outputs from the previous layer and generates inputs to the next layer. An $N$-layer feed-forward neural network eventually maps the input vector $\vec x$ into output $\vec y$, through a composed function $\vec y=f(\vec x):=f^N(f^{N-1}(\dots f^2(f^1(\vec x))\dots))$, where $f^{k\in[1,N]}$ is the feed-forward computation happening at the $k^{th}$ layer. Typically, $f^k$ consists of a linear transformation (defined by a set of weights and a bias) of the output values of the previous layer, and a nonlinear \textit{activation function} applied to the weighted sum. Here we consider piecewise linear ReLU-activation function ($ReLU(\vec x) =\max(0, \vec x)$ element-wise). 

One property class regarding the DNN end-to-end relationship is its robustness against adversarial attacks, i.e., small input perturbations must not cause significant deviations (such as mis-classifications) in the network's output. Proving such a property essentially requires a range propagation from inputs to output. For example, \texttt{ReluVal} \cite{wang2018formal} uses interval arithmetic to perform symbolic interval forward propagation through DNN layers. The resulting range bounds are over-approximation because of lack of capturing dependency with other neurons in the same layer. We have employed our idea of QE-based strongest system-property computation for deriving the precise ranges of ReLU-DNN neurons \cite{ren2019using}. 

\begin{example}\rm
\label{ex6}
For a ReLU-DNN of $N$ layers, let the DNN encoding, regarding its output $y$, be denoted $Encoding_{net}$ that depends on weights and biases of the individual neurons and their activation function and connectivity, be the input-output predicate over $\vec x$ and $y$. Let $R_{\vec x}$ be the predicate of input range constraints. Then the range of an output $y$, denoted $R_{y}$, is obtained as:
\begin{equation}
\label{eqn:07}
R_{y}:=\exists\vec x\big(R_{\vec x}\wedge Encoding_{net}\big).
\end{equation}
\end{example}

Note that, \eqref{eqn:07} is an instance of the strongest system-property formulation~\eqref{eqn:01}. \cite{ren2019using} provides a layer-by-layer forward range propagation computation where each layer has its own formulation of the type~\eqref{eqn:07}), along with a set of on-the-fly QE resolution, to propagate the precise range of $y$ for the ACAS Xu networks \cite{julian2016policy} (of 45 fully-connected ReLU-DNNs, each of which possesses 300 neurons evenly distributed in 6 hidden layers). Owing to QE-based strongest system-property inference, we are able to forward propagate the neurons ranges precisely, allowing for a fine-grained quantitative adversarial analysis.

\section{Conclusion}
\label{sec:conclusion}
Compositional reasoning is a key to scalable approaches for establishing correctness of  model-based designs. In our work, we showed that the foundational principle underlying the compositional reasoning is Quantifier Elimination (QE), which serves as a {\em calculus for composition}. In our compositional reasoning framework,  QE was used to derive the strongest system-property from the given components atomic-properties and their interconnection relation. This converts the system property verification into two fixed steps, independent of the number of components. This also reduces the number of variables needed to verify the system-level postulated properties, therefore improving the efficiency of compositional reasoning. 

We further extended our compositional reasoning framework to support the time-dependent property composition. This is a foundational step towards compositional reasoning of systems with ``memory''. The extension developed a method to determine the system order, given the component orders, and required the time-shifted replicas of component properties during the composition.

We implemented our compositional reasoning paradigm in our prototype verifier \texttt{ReLIC}. It uses the front-end interface of \texttt{AGREE} for model input and composition output, and \texttt{Redlog} at the back-end for performing QE. Within this tool, we also implemented the QE procedure to derive the strongest system-property that can incorporate time-dependence. The proof of a system-level postulated property involves induction, that is also supported in \texttt{ReLIC}.  Our QE-based compositional reasoning has been successfully adopted by an industrial V\&V tool suite from Honeywell Aerospace, and has also been applied for precise range propagation of Simulink models and DNNs for evaluating robustness against adversarial attacks, showing its importance in real world applications.

Aside from QE-based property composition, we noted that a property checking step is also an instance of QE. Thereby, the QE tools provide options to also extend the existing verifiers. In order to take advantage of QE's ability of solving the satisfiability problem over FOL, we integrated the QE tool \texttt{Redlog} as a back-end solver, in parallel to other existing SMT solvers, in the $k$-induction based model-checker \texttt{JKind}, enhancing the latter's capability of model-checking nonlinear properties. In particular, our QE-integrated \texttt{JKind} was able to resolve a fuzzy logic problem involving non-linear computation efficiently, whereas the SMT-integrated version was unable to terminate, showing the extended capability for model-checking, afforded by QE.

Opportunities exist to enhance our tool further to support state-transition representation. Also a wider variety of data-types may be supported in property representation and composition. In particular, the existentially quantified formula of the type given in the paper can be used to express the strongest system-property even for the case when the component atomic-properties are liveness properties. For example, $eventually(out>in)$ can be written as $\bigvee_{k\ge0} (out_k>in_k)$. So if we have a cascade connection of two of same such systems, then its strongest system-property, based on our approach is given by $\exists out_{k\ge0}, out_{j\ge0}\big[\{\bigvee_{k\ge0}(out_k>in_k)\}\wedge\{\bigvee_{j\ge0}(out'_j>out_j)\}\big]$, where $out'$  is the output of the second system. Identifying the classes of liveness properties that admit quantifier elimination is a question of interest of its own right.

\appendices
\section*{Acknowledgment}
The authors acknowledge the helpful inputs from the Rockwell Collins developers Andrew Gacek, John Backes, and Lucas Wagner on their tools \texttt{JKind} and \texttt{AGREE}. Andrew also suggested using Antlr for generating a parser for SMT Lib 2.0.  The fuzzy logic example of Section 2 was provided by Timothy Arnett of University of Cincinnati. \texttt{Redlog} tool was obtained from their site: \url{http://www.redlog.eu/get-redlog/}.

\ifCLASSOPTIONcaptionsoff
  \newpage
\fi



%\bibliographystyle{IEEEtran}
\bibliographystyle{unsrt}     % basic style, author-year citations
\bibliography{ref} 

\end{document}